\documentclass[10pt,journal,compsoc]{IEEEtran}
  \usepackage{cite}
\ifCLASSINFOpdf
\else
 \usepackage[dvips]{graphicx}
\fi
\ifCLASSOPTIONcompsoc
 \usepackage[caption=false,font=footnotesize,labelfont=sf,textfont=sf]{subfig}
\else
 \usepackage[caption=false,font=footnotesize]{subfig}
\fi

\usepackage[cmex10]{amsmath}
\usepackage{amsthm,amssymb}

\newtheorem{theorem}{Theorem}

\newtheorem{corollary}{Corollary}
\newtheorem{remark}{Remark}
\newtheorem{identity}{Identity}

\DeclareMathOperator{\arcsinh}{arcsinh}

\begin{document}
%
\title{Friendship-based Cooperative Jamming for Secure Communication in Poisson Networks}
%
%
%
%

\author{Yuanyu~Zhang,
        Yulong~Shen,
				Hua Wang
			  and~Xiaohong~Jiang,~\IEEEmembership{Senior~Member,~IEEE}
\IEEEcompsocitemizethanks{\IEEEcompsocthanksitem Y.~Zhang and X.~Jiang are with the School of Systems Information Science,
        Future University Hakodate, 116-2, Kameda Nakano-Cho,
        Hakodate, Hokkaido, 041-8655, Japan, and the School of Computer Science and Technology,
        Xidian University, Shaanxi 710071, China. E-mail:yy90zhang@gmail.com;jiang@fun.ac.jp.
\IEEEcompsocthanksitem Y.~Shen is with the School of Computer Science and Technology, Xidian University, Shaanxi 710071, China. E-mail:ylshen@mail.xidian.edu.cn.
\IEEEcompsocthanksitem H.~Wang is with the	Centre of Applied Informatics, College of Engineering and Science, Victoria University, Australia. Email: hua.wang@vu.edu.au.}}
\IEEEtitleabstractindextext{%
\begin{abstract}
Wireless networks with the consideration of social relationships among network nodes are highly appealing for lots of important data communication services. Ensuring the security of such networks is of great importance to facilitate their applications in supporting future social-based services with strong security guarantee. This paper explores the physical layer security-based secure communication in a finite Poisson network with social friendships among nodes, for which a social friendship-based cooperative jamming scheme is proposed. The jamming scheme consists of a Local Friendship Circle (LFC) and a Long-range Friendship Annulus (LFA), where all legitimate nodes in the LFC serve as jammers, but the legitimate nodes in the LFA are selected as jammers through three location-based policies. To understand both the security and reliability performance of the proposed jamming scheme, we first model the sum interference at any location in the network by deriving its Laplace transform under two typical path loss scenarios. With the help of the interference Laplace transform results, we then derive the exact expression for the transmission outage probability (TOP) and determine both the upper and lower bounds on the secrecy outage probability (SOP), such that the overall outage performances of the proposed jamming scheme can be depicted. Finally, we present extensive numerical results to validate the theoretical analysis of TOP and SOP and also to illustrate the impacts of the friendship-based cooperative jamming on the network performances.
\end{abstract}

}

\maketitle

\IEEEdisplaynontitleabstractindextext

%
\IEEEpeerreviewmaketitle

\IEEEraisesectionheading{\section{Introduction}\label{sec_1}}

%
%
%
%
\IEEEPARstart{D}{ue} to the rapid proliferation of smartphones, tablets and PDAs, hand-held devices have been an essential integral part of wireless networks. As these devices are usually carried by human beings, wireless networks, such as mobile ad hoc networks \cite{Kayastha2011}, device-to-device (D2D) communications \cite{YRZhang2015}  and delay-tolerant networks \cite{KMWei2014}, exhibit some social behaviors (e.g., friendship) nowadays. Thus, wireless networks with the consideration of social relationships among network nodes are highly appealing for lots of important data communication services, like content distribution, data sharing and data dissemination \cite{FXian2015}. The inherent open nature of wireless medium makes the information exchange over wireless channels susceptible to eavesdropping attacks from unauthorized users, posing a significant threat to the security of wireless networks \cite{YSShiu2011}. As a result, ensuring the security of such networks is of great importance to facilitate their applications in supporting future social-based services with strong security guarantee, like mobile online social application, location-based application and autonomous mobile application \cite{XHLiang2014}.

The traditional solutions to ensure information security are mainly based on cryptography \cite{Stallings2010}, which encrypts the information with secret keys through various kinds of cryptographic protocols. In cryptography, eavesdroppers are assumed to have limited computing power, so even if they captures the ciphertext, they cannot decrypt it without the secret key. However, as the computing power advances rapidly nowadays, these solutions are facing increasingly high risk of being broken by the relentless attempts of eavesdroppers. In addition, due to the lack of centralized control, secret key management and distribution in decentralized wireless networks are very costly and complex to be implemented. This necessitates the introduction of more powerful schemes to ensure wireless network security. Physical layer (PHY) security \cite{bloch2011physical} has been recognized as a promising strategy to provide a strong form of security for wireless communications. The basic principle of PHY security is to exploit the inherent randomness of noise and wireless channels to ensure the confidentiality of messages against any eavesdropper regardless of its computing power \cite{Mukherjee2014}. Compared to the cryptography-based solutions, PHY security can offer some major advantages, like an everlasting security guarantee, no need for key management/distribution, a high scalability for the next-generation networks \cite{NYang2015}. 

Some recent efforts have been devoted to the study of PHY security-based secure communication in wireless networks with social relationships. Wang \emph{et al.} \cite{wang2015secure} considered a D2D communication scenario, where the head of two D2D user (DUE) clusters wish to communicate with the help of an intermediate Decode-and-Forward relay. The communication security is guaranteed by the cooperative jamming scheme, where multiple friendly jammers send jamming signals to suppress eavesdroppers, and the social relationship is modeled by a social trust parameter $\mu\in[0,1]$. Two sets of jammers (one set per cluster) are selected from DUEs with social trust above some threshold $\mu_{min}$. With the consideration of power constraint, the authors studied the optimal selection of relay and jammers to maximize the secrecy rate of DUE transmission and also to ensure a required signal-to-interference-plus-noise ratio (SINR) level to cellular users. Tang \emph{et al.} \cite{LTang2015} considered a wireless network consisting of one source-destination pair, a set of cooperative jammers and one eavesdropper. Cooperative jamming is adopted to ensure the security and the concept of social tie is introduced to model the social relationship between jammers and the source/destination. The strength of social tie of the $n$-th jammer is denoted by $a_n\in\{0,1\}$, where $1$ ($0$) indicates that the jammer is (is not) willing to participate in the cooperative jamming. The authors modeled the decision problem of jammers as a social tie-based cooperative jamming game and then explored the secrecy outage performance of the source-destination pair by computing the Nash equilibrium of the game.

While the above works represent a significant process in the study of PHY security-based secure communication in wireless networks with social relationships, the social relationships they considered are simply modeled by an indicator variable. Although these variables are acceptable for characterizing some location-independent social relationships, like social tie and social trust, they may fail to model some important social properties closely related to geometric properties of networks, e.g., small-world phenomenon \cite{Kleinberg2000,inaltekin2014delay}. Also, the network scenarios they considered are quite simple, which consists of either only one eavesdropper and several jammers or only two clusters of jammers. To the best of our knowledge, the study of PHY security-based secure communication in more general large scale wireless networks with small-world social relationships still remains unknown, which is the scope of this paper.

This paper considers a finite Poisson network consisting of one transmitter-receiver pair, multiple legitimate nodes and multiple eavesdroppers distributed according to two independent and homogeneous Poisson Point Processes (PPP), respectively. It is notable that the Poisson network model can nicely capture the random geometric properties of networks and enable the analytical modeling of network interference statistics in general \cite{Haenggi2009}, so it has been widely used in the PHY security performance study of large scale wireless networks without the consideration of social relationships \cite{Pinto2012PartI, Pinto2012PartII,Rabbachin2015,Zhou2011,HWang2013,Geraci2014Cellular,ChMa2015, Xu2016,YWLiu2015ICC} (Please refer to Section \ref{sec_6} for related works). In particular, we consider a more realistic location-based friendship model to characterize the small-world social relationships among nodes in the network. The main contributions of this paper are summarized as follows.

\begin{itemize}
\item This paper proposes a friendship-based cooperative jamming scheme to ensure the PHY security-based secure communication between the transmitter and receiver. The jamming scheme comprises a Local Friendship Circle (LFC) and a Long-range Friendship Annulus (LFA), where all legitimate nodes in the LFC serve as jammers, and three location-based policies are designed to select legitimate nodes in the LFA as jammers.
\item The transmission outage probability (TOP) and secrecy outage probability (SOP) are adopted to model the reliability and security performance of the proposed jamming scheme \cite{XYZhou2011}. For the modeling of these performance metrics, we first conduct analysis of the sum interference at any location in the network by deriving its Laplace transforms under the three location-based jammer selection policies and two typical path loss scenarios \cite{Rappaport2001}.

\item With the help of the interference Laplace transform results, we then derive the exact expression for the TOP and determine both the upper and lower bounds on the SOP, such that the overall outage performances of the proposed jamming scheme can be fully depicted. 
\item 
Finally, we present extensive numerical results to validate the theoretical analysis of TOP and SOP and also to illustrate the impacts of the friendship-based cooperative jamming on the network performance.
\end{itemize}

The remainder of this paper is organized as follows. Section \ref{sec_2} introduces the preliminaries and friendship-based cooperative jamming scheme. The Laplace transforms of the sum interference are analyzed in Section \ref{sec_3} and the TOP and SOP are analyzed in Section \ref{sec_4}. The numerical results and corresponding discussions are provided in Section \ref{sec_5}. Section \ref{sec_6} presents the related works of PHY security performance study for Poisson networks without social relationships. Finally, we conclude this paper in Section \ref{sec_7}.

\section{Preliminaries and Jamming Scheme} \label{sec_2} 

\subsection{System Model}
\begin{figure}[!t]
\centering
\includegraphics[width=2.5in]{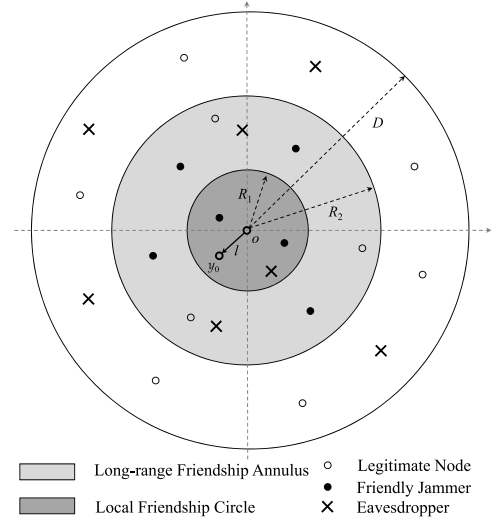}
\caption{System model: nodes are distributed over a bi-dimensional disk $\mathcal B(o,D)$ with radius $D$. The transmitter is located at the origin $o$ and the receiver is located at $y_0$ with $||y_0||=l$. Legitimate nodes and eavesdroppers are distributed according to two independent homogeneous PPPs. The friendship-based cooperative jamming model comprises a LFC with radius $R_1$ and a LFA with inner radius $R_1$ and outer radius $R_2$.}
\label{fig_sysmodel}
\end{figure}

As illustrated in Fig.\ref{fig_sysmodel}, we consider a finite wireless network with nodes distributed over a bi-dimensional disk $\mathcal B(o,D)\subset \mathbb R^2$ with radius $D$. The network consists of a transmitter located at the origin $o$ and a receiver located at $y_0$ with fixed distance $||y_0||=l$ to $o$. Also present in the network are multiple legitimate nodes and multiple eavesdroppers, whose locations are modeled as two independent and homogeneous PPPs $\Phi$ and $\Phi_E$ with intensities $\lambda$ and $\lambda_e$, respectively. Throughout this paper we will use $x$ ($z$) to denote the random location of a legitimate node (eavesdropper) as well as the node (eavesdropper) itself. To suppress the eavesdroppers, a set of legitimate nodes will serve as jammers to send jamming signals. The set of jammer locations is denoted by $\Phi_J$.

The channel suffers from both small-scale Rayleigh fading and large-scale log-distance path loss with exponent $\alpha\geq 2$ \cite{Rappaport2001}. The fading coefficient is constant for a block of transmission and varies randomly and independently from block to block for all channels. We assume that the transmitter and jammers transmit with the same power. Without loss of generality, unit transmit power is assumed. The sum interference caused by the set of jammers at any location $y$ in the network is then given by 
\begin{IEEEeqnarray}{rCl}
I(y)=\sum_{x\in\Phi_J}h_{x,y}||x-y||^{-\alpha},
\end{IEEEeqnarray}
where $h_{x,y}$ is the fading coefficient between $x$ and $y$, and $||x-y||$ is the distance between $x$ and $y$. Due to the Rayleigh fading assumption, $h_{x,y}$ is exponentially distributed. We assume unit mean for $h_{x,y}$, i.e., $\mathbb E[h_{x,y}]=1$. The network is assumed interference-limited, and hence, the ambient noise is negligible. The signal-to-interference ratio (SIR) for the receiver $y_0$ from the transmitter $o$ is then given by
\begin{IEEEeqnarray}{rCl}
\mathrm{SIR}_{y_0} = \frac{h_{o,y_0}l^{-\alpha}}{I(y_0)},
\end{IEEEeqnarray}
and the SIR for any eavesdropper $z\in \Phi_E$ is given by
\begin{IEEEeqnarray}{rCl}
\mathrm{SIR}_{z} = \frac{h_{o,z}||z||^{-\alpha}}{I(z)}.
\end{IEEEeqnarray}

\subsection{Friendship-based Cooperative Jamming}\label{sec_2_2}

To ensure the transmission security, this paper proposes a friendship-based cooperative jamming scheme by exploiting the inherent friendship between the transmitter and legitimate nodes. In this scheme, only the legitimate nodes that are friends of 
the transmitter serve as jammers. It was demonstrated in \cite{inaltekin2014delay} that each node has not only local friends in a circle around itself but also $N$ long-range friends randomly selected from the region outside the local circle. It is notable that $N$ can be drawn from any given discrete probability distribution. 

Based on the model in \cite{inaltekin2014delay}, the proposed jamming scheme is composed of a Local Friendship Circle (LFC) with radius $R_1$ and a Long-range Friendship Annulus (LFA) with inner radius $R_1$ and outer radius $R_2$, where $0<R_1\leq R_2\leq D$ (illustrated in Fig.\ref{fig_sysmodel}). Both the LFC and LFA are centered at the transmitter (i.e., the origin $o$). Let $\mathcal A_1$ denote the LFC and $\mathcal A_2$ denote the LFA. In the proposed jamming scheme, all legitimate nodes in $\mathcal A_1$ serve as jammers, while each legitimate node $x$ in $\mathcal A_2$ is selected as a jammer through a location-based policy $P(||x||)\in[0,1]$. Notice that different $P(||x||)$ can yield different distributions of long-range jammers (i.e., different $\Phi_J$). In this paper, we design three selection policies $P(||x||)$, which are summarized as follows.

\begin{itemize}
\item \textbf{Policy E}: For each node $x\in\Phi\cap \mathcal A_2$, $P(||x||)=p$, where $p\in[0,1]$. This policy corresponds to the scenario where long-range jammers are uniformly distributed over $\mathcal A_2$.
\item \textbf{Policy I}: For each node $x\in\Phi\cap \mathcal A_2$, $P(||x||)$ is increasing with its path loss to the transmitter, i.e., 
\begin{IEEEeqnarray}{rCl}\label{eqn_policyI}
P(||x||)=\frac{||x||^\alpha-R_1^{\alpha}}{R_2^\alpha-R_1^\alpha}.
\end{IEEEeqnarray}
 This policy corresponds to the scenario where most of the long-range jammers are distributed near $R_2$.
\item \textbf{Policy D}: For each node $x\in\Phi\cap \mathcal A_2$, $P(||x||)$ is decreasing with its path loss to the transmitter, i.e., 
\begin{IEEEeqnarray}{rCl}\label{eqn_policyD}
P(||x||)=\frac{R_2^{\alpha}-||x||^\alpha}{R_2^\alpha-R_1^\alpha}.
\end{IEEEeqnarray}
This policy corresponds to the scenario where most of the long-range jammers are distributed near $R_1$.
\end{itemize}  

\begin{remark}
The policy $P(||x||)$ can be interpreted as a thinning operation on $\Phi$ \cite{Chiu2013}. According to the property of thinning operation, the number of jammers in $\mathcal A_2$ still follows a Poisson distribution. Hence, the friendship model in the proposed jamming scheme is a special case of the one in \cite{inaltekin2014delay}, given that $N$ is drawn from a Poisson distribution. Also, from (\ref{eqn_policyI}) and (\ref{eqn_policyD}), we can see that Policy $\mathrm{D}$ generates more long-range jammers than Policy $\mathrm{I}$.
\end{remark}
\subsection{Performance Metrics}
The impact of friendship-based cooperative jamming scheme on the communication between the transmitter $o$ and receiver $y_0$ is two-edged. On one hand, the interference generated by the jammers can degrade the eavesdropper channels, which may greatly enhance the security of the communication. On the other hand, the transmitter-receiver link is also impaired by the unintended interference, resulting in a probably unreliable communication. In this paper, we will adopt the concepts of \emph{transmission outage probability} (TOP) and \emph{secrecy outage probability} (SOP) to measure the reliability and security of the transmitter-receiver communication \cite{XYZhou2011}, which can be defined according to  the following outage events.

\begin{itemize}
\item \textbf{Transmission outage}: The SIR at the receiver $y_0$ is below some threshold $\beta$, i.e., $\mathrm{SIR}_{y_0}<\beta$, which results in that the receiver $y_0$ fails to decode the message from the transmitter $o$. The probability that this event happens is referred to as the TOP. 
\item \textbf{Secrecy outage}: The  SIR at one or more eavesdroppers is above some threshold $\beta_e$, which results in that the eavesdroppers can intercept the message from the transmitter $o$. The probability that this event happens is referred to as the SOP.
\end{itemize}
Formally, the TOP is given by
\begin{IEEEeqnarray}{rCl}\label{eqn_df_top}
p_{to}=\mathbb P(\mathrm{SIR}_{y_0}<\beta),
\end{IEEEeqnarray}
and the SOP is given by
\begin{IEEEeqnarray}{rCl}\label{eqn_df_sop}
p_{so}=\mathbb P\left(\bigcup_{z\in\Phi_E}\mathrm{SIR}_{y_0}>\beta_e\right).
\end{IEEEeqnarray}

\section{Laplace Transform of the Sum Interference}\label{sec_3}

In this section, the Laplace transform of the sum interference $I(y)$ at any location $y\in\mathcal B(o,D)$ is analyzed for all three long-range jammer selection policies. To make the analysis mathematically tractable, we focus on two typical path loss scenarios of $\alpha=2$ and $\alpha=4$.

According to the definition, the Laplace transform of $I(y)$ is given by  
\begin{IEEEeqnarray}{rCl}
\mathcal L_{I(y)}^{\Xi,\alpha}(s)&=&\mathbb E_{I(y)}\left[e^{-sI(y)}\right]\nonumber\\
&=&\mathbb E_{\Phi_J,\left\{h_{x,y}\right\}}\left[\mathrm{exp}\left(-s\sum_{x\in\Phi_J}h_{x,y}||x-y||^{-\alpha}\right)\right]\nonumber\\
&=&\mathbb E_{\Phi_J,\left\{h_{x,y}\right\}}\left[\prod_{x\in\Phi_J}\mathrm{exp}\left(-sh_{x,y}||x-y||^{-\alpha}\right)\right]\nonumber\\
&=&\mathbb E_{\Phi_J}\left[\prod_{x\in\Phi_J}\mathbb E_{h}\left[\mathrm{exp}\left(-sh||x-y||^{-\alpha}\right)\right]\right]\nonumber\\
&=&\mathbb E_{\Phi_J}\left[\prod_{x\in\Phi_J}\frac{1}{1+s||x-y||^{-\alpha}}\right],
\end{IEEEeqnarray}
where $\Xi=\mathrm{E, I, D}$ denotes the selection policy.

From the cooperative jamming scheme in Section \ref{sec_2_2}, we can see that $\Phi_J$ is indeed an inhomogeneous PPP obtained by applying two independent thinning operations on $\Phi$.
We now define the intensity measure of $\Phi_J$ by $\mathit{\Lambda}(\cdot)$, which gives the expected number of nodes in a given set. By applying the probability generating functional of $\Phi_J$, we have

\begin{IEEEeqnarray}{rCl}\label{eqn_LT_general}
\mathcal L_{I(y)}^{\Xi,\alpha}(s)&=&\mathrm{exp}\left\{-\int_{\mathcal B(o,D)}\left(1-\frac{1}{1+s||x-y||^{-\alpha}}\right)\mathit{\Lambda}(\mathrm{d}x)\right\}\nonumber \\
&=&\mathrm{exp}\left\{-\underbrace{\int_{\mathcal B(o,D)}\left(\frac{s}{s+||x-y||^{\alpha}}\right)\mathit{\Lambda}(\mathrm{d}x)}_A\right\},
\end{IEEEeqnarray}
where $\mathit{\Lambda}(\mathrm{d}x)$ is given by
\begin{IEEEeqnarray}{rCl}
\mathit{\Lambda}(\mathrm{d}x)=\left\{\begin{matrix}
\lambda \mathrm{d}x, & x\in\mathcal A_1\\ 
\lambda P(||x||)\mathrm{d}x, & x\in\mathcal A_2.
\end{matrix}\right.,
\end{IEEEeqnarray}
following from the thinning property of PPP. The term $A$ in (\ref{eqn_LT_general}) can be rewritten as 
\begin{IEEEeqnarray}{rCl}\label{eqn_A}
A&=&\lambda\underbrace{\int_{\mathcal A_1}\left(\frac{s}{s+||x-y||^{\alpha}}\right)\mathrm{d}x}_{B_\alpha}\nonumber\\
&&+\lambda\underbrace{\int_{\mathcal A_2}\left(\frac{s}{s+||x-y||^{\alpha}}\right)P(||x||)\mathrm{d}x}_{C_\alpha }.
\end{IEEEeqnarray}
Changing Cartesian coordinates to polar coordinates, we can rewrite $B_\alpha$ and $C_\alpha$ in (\ref{eqn_A}) as

\begin{IEEEeqnarray}{rCl}\label{eqn_B_alpha}
B_\alpha&=&2\int_{0}^{R_1}\int_{0}^{\pi}\frac{sr\mathrm{d}\theta\mathrm{d}r}{s+(r^2+||y||^2-2r||y||\cos\theta)^{\alpha/2}},
\end{IEEEeqnarray}
and
\begin{IEEEeqnarray}{rCl}\label{eqn_C_alpha}
C_\alpha&=&2\int_{R_1}^{R_2}\int_{0}^{\pi}\frac{srP(r)\mathrm{d}\theta\mathrm{d}r}{s+(r^2+||y||^2-2r||y||\cos\theta)^{\alpha/2}}.
\end{IEEEeqnarray}

\subsection{ The Case of $\alpha = 2$} 

In this subsection, we derive the Laplace transform of $I(y)$ for the case of $\alpha=2$.  The main results are summarized in the following theorem.

\begin{theorem}\label{theorem_LT_2}
For the case of $\alpha=2$, the Laplace transform of the sum interference $I(y)$ at any location $y\in \mathcal B(o,D)$ under Policy $\mathrm{E}$  is given by 
\begin{IEEEeqnarray}{rCl}\label{eqn_LT_E2}
\mathcal L_{I(y)}^{\mathrm{E},2}(s)&=&\mathrm{exp}\Bigg\{-\lambda\pi s\bigg[p \arcsinh \frac{s+R_2^2-||y||^2}{2||y||\sqrt{s}}\\
&&+(1-p)\arcsinh \frac{s+R_1^2-||y||^2}{2||y||\sqrt{s}}-\ln\frac{\sqrt{s}}{||y||}\bigg]\Bigg\},\nonumber
\end{IEEEeqnarray}
where $\arcsinh t=\mathrm{ln}(t+\sqrt{t^2+1})$ denotes the inverse hyperbolic sine function.
The Laplace transform of $I(y)$ under Policy $\mathrm{I}$ and Policy $\mathrm{D}$ is given by
\begin{IEEEeqnarray}{rCl}\label{eqn_LT_I_D2}
\mathcal L_{I(y)}^{\Xi',2}(s)\nonumber&=&\mathrm{exp}\Bigg\{-\lambda\pi s\bigg[\Psi_2^{\Xi'}(R_2,s,||y||)-\Psi_{2}^{\Xi'}(R_1,s,||y||)\nonumber\\
&&+\bigg(\arcsinh \frac{s+R_1^2-||y||^2}{2||y||\sqrt{s}}-\ln\frac{\sqrt{s}}{||y||}\bigg)\bigg]\Bigg\},
\end{IEEEeqnarray}
where
$\Xi'=\mathrm{I}$ and $\mathrm{D}$, 
\begin{IEEEeqnarray}{rCl}
\Psi_2^{\mathrm{I}}(r,s,||y||)&=&\frac{ \sqrt{(r^4+2(s-||y||^2)r^2+(s+||y||^2)^2}}{R_2^2-R_1^2}\nonumber\\
&&-\frac{s+R_1^2-||y||^2}{R_2^2-R_1^2}\arcsinh \frac{s+r^2-||y||^2}{2||y||\sqrt{s}},\nonumber
\end{IEEEeqnarray}
and 
\begin{IEEEeqnarray}{rCl}
\Psi_2^{\mathrm{D}}(r,s,||y||)&=&\frac{s+R_2^2-||y||^2}{R_2^2-R_1^2}\arcsinh \frac{s+r^2-||y||^2}{2||y||\sqrt{s}}\nonumber\\
&&-\frac{\sqrt{(r^4+2(s-||y||^2)r^2+(s+||y||^2)^2}}{R_2^2-R_1^2}.\nonumber 
\end{IEEEeqnarray}
\end{theorem} 

\begin{proof}
The proof is given in Appendix \ref{sec_app2}.
\end{proof}

\subsection{The Case of $\alpha=4$}
The Laplace transform of $I(y)$ for the case of $\alpha=4$ is derived in this subsection. The main results are summarized in the following theorem. 
\begin{theorem}\label{theorem_LT_4}
For the case of $\alpha=4$, the Laplace transform of the sum interference $I(y)$ at any location $y\in \mathcal B(o,D)$ under Policy $\mathrm{E}$  is given by
\begin{IEEEeqnarray}{rCl}\label{eqn_LT_E4}
\mathcal L_{I(y)}^{\mathrm{E},4}(s)&=&\mathrm{exp}\Bigg\{-\lambda\pi\sqrt{s}\bigg[\frac{\pi}{2}-(1-p)\\
&&\times\arctan\frac{\sqrt{s}+\psi(R_1,s,||y||)}{\eta(R_1,s,||y||)+R_1^2-||y||^2}\nonumber\\
&&-p\arctan\frac{\sqrt{s}+\psi(R_2,s,||y||)}{\eta(R_2,s,||y||)+R_2^2-||y||^2}\bigg]\Bigg\},\nonumber
\end{IEEEeqnarray}
where \begin{IEEEeqnarray}{rCl}\label{eqn_eta}
&&\eta(r,s,||y||)\\
&&=\frac{\sqrt{\sqrt{(g(r,s,||y||))^2+4s(r^2+||y||^2)^2}+g(r,s,||y||)}}{\sqrt{2}},\nonumber
\end{IEEEeqnarray}  
\begin{IEEEeqnarray}{rCl}
g(r,s,||y||)=(r^2-||y||^2)^2-s,
\end{IEEEeqnarray} 
\begin{IEEEeqnarray}{rCl}\label{eqn_psi}
\psi(r,s,||y||)&=&\frac{\sqrt{s}(r^2+||y||^2)}{\eta(r,s,||y||)},
\end{IEEEeqnarray} 
and $\arctan t$ is the inverse tangent function. The Laplace transform of $I(y)$ under Policy $\mathrm{I}$ and Policy $\mathrm{D}$ is given by
\begin{IEEEeqnarray}{rCl}\label{eqn_LT_I_D4}
\mathcal L_{I(y)}^{\Xi',4}(s)&=&\mathrm{exp}\Bigg\{-\lambda\pi\sqrt{s}\nonumber\\
&&\times\bigg[\frac{\pi}{2}-\arctan\frac{\sqrt{s}+\psi(R_1,s,||y||)}{\eta(R_1,s,||y||)+R_1^2-||y||^2}\nonumber\\
&&+\Psi_{4}^{\Xi'}(R_2,s,||y||)-\Psi_{4}^{\Xi'}(R_1,s,||y||)\bigg]\Bigg\},
\end{IEEEeqnarray}
where $\Xi'=\mathrm{I}$ and $\mathrm{D}$, 
\begin{IEEEeqnarray}{rCl}
\Psi_{4}^{\mathrm{I}}(r,s,||y||)&=&\frac{2\sqrt{s}||y||^2}{R_2^4-R_1^4}\ln\bigg[(\eta(r,s,||y||)+r^2-||y||^2)^2\nonumber\\
&&+(\sqrt{s}+\psi(r,s,||y||))^2\bigg]-\frac{1}{2(R_2^4-R_1^4)}\nonumber\\
&&\times\bigg[(r^2+3||y||^2)\psi(r,s,||y||)\nonumber\\
&&-3\sqrt{s}\eta(r,s,||y||)\bigg]+\frac{s+R_1^4-||y||^4}{R_2^4-R_1^4}\nonumber\\
&&\times\arctan\frac{\sqrt{s}+\psi(r,s,||y||)}{\eta(r,s,||y||)+r^2-||y||^2},
\end{IEEEeqnarray}
and
\begin{IEEEeqnarray}{rCl}
\Psi_{4}^{\mathrm{D}}(r,s,||y||)&=&-\frac{2\sqrt{s}||y||^2}{R_2^4-R_1^4}\ln\bigg[(\eta(r,s,||y||)+r^2-||y||^2)^2\nonumber\\
&&+(\sqrt{s}+\psi(r,s,||y||))^2\bigg]+\frac{1}{2(R_2^4-R_1^4)}\nonumber\\
&&\times\bigg[(r^2+3||y||^2)\psi(r,s,||y||)\nonumber\\
&&-3\sqrt{s}\eta(r,s,||y||)\bigg]-\frac{s+R_2^4-||y||^4}{R_2^4-R_1^4}\nonumber\\
&&\times\arctan\frac{\sqrt{s}+\psi(r,s,||y||)}{\eta(r,s,||y||)+r^2-||y||^2}.
\end{IEEEeqnarray}
\end{theorem}

\begin{proof}
The proof is given in Appendix \ref{sec_app3}.
\end{proof}

\begin{corollary}
For $P(r)=0$, as $R_1\rightarrow \infty$, the Laplace transform of $I(y)$ for the case of $\alpha=4$ is
\begin{IEEEeqnarray}{rCl}
\mathcal L_{I(y)}^{\Xi, 4}(s)&=& \mathrm{exp}\left(- \frac{\lambda\sqrt{s}\pi^2}{2}\right),
\end{IEEEeqnarray}
which recovers the well-known Laplace transform of $I(y)$ for a homogeneous infinite PPP with $\alpha=4$ \cite{Haenggi2009}.
\end{corollary}
\begin{proof}
Letting $P(r)=0$ yields 
\begin{IEEEeqnarray}{rCl}
\mathcal L_{I(y)}^{\Xi,4}(s)&=&\mathrm{exp}\Bigg\{-\lambda\pi\sqrt{s}\\
&&\bigg[\frac{\pi}{2}-\arctan\frac{\sqrt{s}+\psi(R_1,s,||y||)}{\eta(R_1,s,||y||)+R_1^2-||y||^2}\bigg]\Bigg\}\nonumber.
\end{IEEEeqnarray}
As $R_1\rightarrow \infty$,
\begin{IEEEeqnarray}{rCl}
&&\lim_{R_1\rightarrow \infty}\arctan\frac{\sqrt{s}+\psi(R_1,s,||y||)}{\eta(R_1,s,||y||)+R_1^2-||y||^2}\nonumber\\
&&=\arctan \frac{2\sqrt{s}}{\infty-||y||^2}\nonumber\\
&&=0,
\end{IEEEeqnarray}
which completes the proof.
\end{proof}

\section{Outage Performance}\label{sec_4}
In this section, the TOP and SOP of the proposed cooperative jamming scheme are analyzed. Similar to Section \ref{sec_3}, we focus again on the cases of $\alpha=2$ and $\alpha=4$. The analysis is based on the Laplace transforms of the sum interference $I(y)$ derived in Section \ref{sec_3}. We first determine the exact expression for the TOP and then obtain both the upper and lower bounds on the SOP.

\subsection{Transmission Outage Probability}
The TOP can be regarded as a measure of the link reliability between the transmitter $o$ and receiver $y_0$. For the Rayleigh fading channel model, the TOP can be directly derived by applying the Laplace transform of the sum interference at the receiver $y_0$ \cite{Haenggi2009}. The following theorem is established to summarize the result of the TOP.
\begin{theorem}\label{theorem_TOP}
Consider a finite Poisson network with nodes distributed over a bi-dimensional disk $\mathcal B(o,D)$ as illustrated in Fig.\ref{fig_sysmodel} and the friendship-based cooperative jamming scheme in Section \ref{sec_2_2}, the TOP of the transmitter-receiver pair is given by
\begin{IEEEeqnarray}{rCl}
p_{to}&&=1-\mathcal L_{I(y_0)}^{\mathrm{\Xi},\alpha}(\beta l^{\alpha}),
\end{IEEEeqnarray}
where $\Xi=\mathrm{E}$, $\mathrm{I}$ and $\mathrm{D}$ denotes the long-range jammer selection policy, $\alpha$ denotes the path loss exponent, and the Laplace transform $\mathcal L_{I(y_0)}^{\Xi,\alpha}(\beta l^{\alpha})$ of the sum interference at the receiver $y_0$ is  given by (\ref{eqn_LT_E2}), (\ref{eqn_LT_I_D2}), (\ref{eqn_LT_E4}), (\ref{eqn_LT_I_D4}) with $||y_0||=l$, $s = \beta l^{\alpha}$ for the cases of $\alpha=2$ and $\alpha=4$, respectively.
\end{theorem}

\begin{proof}
From the definition of TOP in (\ref{eqn_df_top}), we have 
\begin{IEEEeqnarray}{rCl}
p_{to}&=&\mathbb P\left(\mathrm{SIR}_{y_0}<\beta\right)\nonumber\\
&=&\mathbb P\left(\frac{h_{o,y_0}l^{-\alpha}}{I(y_0)}<\beta\right)\nonumber\\
&=&\mathbb E_{\Phi_J}\left[\mathbb P\left(\frac{h_{o,y_0}l^{-\alpha}}{I(y_0)}<\beta\big|\Phi_J\right)\right]\nonumber\\
&=&\mathbb E_{\Phi_J}\left[\mathbb P\left(h_{o,y_0}<\beta l^{\alpha} I(y_0)\big|\Phi_J\right)\right]\nonumber\\
&=&1-\mathbb E_{I(y_0)}\left[e^{-\beta l^{\alpha} I(y_0)}\right]\nonumber\\
&=&1- \mathcal L_{I(y_0)}^{\mathrm{\Xi},\alpha}(\beta l^{\alpha}),
\end{IEEEeqnarray}
which completes the proof.
\end{proof}

\subsection{Secrecy Outage Probability}
The SOP is a commonly-used performance metric to quantify the PHY security. In the performance analysis of large-scale systems, the exact SOP is usually unavailable, mainly due to the reason that the analysis involves computing highly cumbersome integrals in terms of the PPPs of both legitimate nodes and eavesdroppers. We therefore resort to obtain the upper and lower bounds on the SOP by applying the bounding technique used in \cite{Zhou2011}. We establish the following theorem to summarize the main results.
\begin{theorem}
Consider a finite Poisson network with nodes distributed over a bi-dimensional disk $\mathcal B(o,D)$ as illustrated in Fig.\ref{fig_sysmodel} and the friendship-based cooperative jamming scheme in Section \ref{sec_2_2}, the upper bound on the SOP of the transmitter-receiver pair is given by
\begin{IEEEeqnarray}{rCl}
p_{so}^{\mathrm{UB}}=1-\mathrm{exp}\left\{-2\pi\lambda_e\int_{0}^{D}\mathcal L_{I(z)}^{\Xi,\alpha}(\beta_er_e^\alpha)r_e\mathrm{d}r_e\right\},
\end{IEEEeqnarray}
and the lower bound is given by 
\begin{IEEEeqnarray}{rCl}
p_{so}^{\mathrm{LB}}=\int_{0}^{D}2\lambda_e\pi r_{e^*}\mathrm{exp}(-\lambda_e \pi r_{e^*}^2)\mathcal L_{I(z^*)}^{\Xi,\alpha}(\beta_e r_{e^*}^{\alpha})\mathrm{d}r_{e^*},
\end{IEEEeqnarray}
where $\Xi=\mathrm{E}$, $\mathrm{I}$ and $\mathrm{D}$ denotes the long-range jammer selection policy, $\alpha$ denotes the path loss exponent, $z^*$ denotes the eavesdropper nearest to the transmitter $o$, $r_{e^*}$ denotes the distance between $z^*$ and $o$,  and the Laplace transform $\mathcal L_{I(z)}^{\Xi,\alpha}(\beta r_e^{\alpha})$ is given by (\ref{eqn_LT_E2}), (\ref{eqn_LT_I_D2}), (\ref{eqn_LT_E4}), (\ref{eqn_LT_I_D4}) with $||z||=r_e$, $s = \beta r_e^{\alpha}$ for the cases of $\alpha = 2$ and $\alpha = 4$, respectively. 
\end{theorem}
\begin {proof}

From the definition of SOP in (\ref{eqn_df_sop}), we have
\begin{IEEEeqnarray}{rCl}
&&p_{so}=\mathbb P\left(\bigcup_{z\in\Phi_E}\mathrm{SIR}_{y_0}>\beta_e\right)\nonumber\\
&&=1-\mathbb P\left(\bigcap_{z\in \Phi_E}\mathrm{SIR}_z<\beta_e\right)\nonumber\\
&&=1-\mathbb E_{\Phi_J}\left[\mathbb E_{\Phi_E}\left[\mathbb P\left(\bigcap_{z\in \Phi_E}\frac{h_{o,z}||z||^{-\alpha}}{I(z)}<\beta_e\big|\Phi_E,\Phi_J\right)\right]\right]\nonumber\\
&&\overset{(a)}=1-\mathbb E_{\Phi_J}\left[\mathbb E_{\Phi_E}\left[\prod_{z\in \Phi_E}\mathbb P\left(\frac{h_{o,z}||z||^{-\alpha}}{I(z)}<\beta_e\big|\Phi_E,\Phi_J\right)\right]\right]\nonumber\\
&&=1-\mathbb E_{\Phi_J}\Bigg[\mathbb E_{\Phi_E}\Bigg[\nonumber\\
&&~~~\prod_{z\in \Phi_E}\left(1-\mathbb P\left(\frac{h_{o,z}||z||^{-\alpha}}{I(z)}>\beta_e\big|\Phi_E,\Phi_J\right)\right)\Bigg]\Bigg]\nonumber\\
&&\overset{(b)}=1-\mathbb E_{\Phi_J}\Bigg[\mathrm{exp}\Bigg\{\nonumber\\
&&~~~-\lambda_e\int_{\mathcal B(o,D)}\mathbb P\left(\frac{h_{o,z}||z||^{-\alpha}}{I(z)}>\beta_e\big|\Phi_J\right)\mathrm{d}z\Bigg\}\Bigg],
\end{IEEEeqnarray}
where $(a)$ follows since $h_{o,z}$, $z\in\Phi_E$ are i.i.d. random variables, and $(b)$ follows from applying the probability generating functional of $\Phi_E$.
Applying the Jensen's Inequality yields the upper bound on $p_{so}$   
\begin{IEEEeqnarray}{rCl}
p_{so}&\leq&1-\mathrm{exp}\Bigg\{\nonumber\\
&&-\lambda_e\int_{\mathcal B(o,D)}\mathbb E_{\Phi_J}\left[\mathbb P\left(\frac{h_{o,z}||z||^{-\alpha}}{I(z)}>\beta_e\big|\Phi_J\right)\right]\mathrm{d}z\Bigg\}\nonumber\\
&=&1-\mathrm{exp}\left\{-\lambda_e\int_{\mathcal B(o,D)}\mathcal L_{I(z)}^{\mathrm{\Xi},\alpha}(\beta_e||z||^\alpha)\mathrm{d}z\right\}\nonumber\\
&=&1-\mathrm{exp}\left\{-2\pi\lambda_e\int_{0}^{D}\mathcal L_{I(z)}^{\mathrm{\Xi},\alpha}(\beta_er_e^\alpha)r_e\mathrm{d}r_e\right\}.
\end{IEEEeqnarray}

The lower bound is obtained by considering only the eavesdropper $z^*$ nearest to the transmitter $o$. Let $R_{z^*}$ denote the random distance between $z^*$ and $o$. The probability distribution function of $R_{z^*}$ can be given by 
\begin{IEEEeqnarray}{rCl}
f_{R_{z^*}}(r_{e^*})=
\left\{\begin{matrix}
 2\lambda_e\pi r_{e^*}\mathrm{exp}(-\lambda_e \pi r_{e^*}^2), & 0\leq r_{e^*}\leq D\nonumber \\ 
 0,& \mathrm{otherwise} 
\end{matrix}\right..
\end{IEEEeqnarray}
Please refer to Appendix \ref{sec_app4} for the proof. The SOP can then be bounded from below by the probability that $z^*$ causes a secrecy outage, i.e., 
\begin{IEEEeqnarray}{rCl}
p_{so}&\geq&\mathbb P(\mathrm{SIR}_{z^*}>\beta_e)\\
&=&\int_{0}^{D}\mathbb P\left(\frac{h_{o,z^*}r_{e^*}^{-\alpha}}{I(z^*)}>\beta_e\right)f_{R_{z^*}}(r_{e^*})\mathrm{d}r_{z^*}\nonumber\\
&=&\int_{0}^{D}2\lambda_e\pi r_{e^*}\mathrm{exp}(-\lambda_e \pi r_{e^*}^2)\mathcal L_{I(z^*)}^{\mathrm{\Xi},\alpha}(\beta_e r_{e^*}^{\alpha})\mathrm{d}r_{e^*}.\nonumber
\end{IEEEeqnarray}
\end{proof}

\begin{corollary}\label{corollary2}
As the network size tends to infinity, i.e., $D\rightarrow \infty$, the SOP $p_{so}\rightarrow 1$ under all long-range jammer selection policies $\mathrm{E}$, $\mathrm{I}$ and $\mathrm{D}$ for the cases of $\alpha=2$ and $\alpha=4$. 
\end{corollary}
\begin{proof}
See Appendix \ref{sec_app5} for the proof.
\end{proof}

\begin{figure*}[!t]
\centering
\subfloat[TOP for $\alpha=2$]{\includegraphics[width=3in]{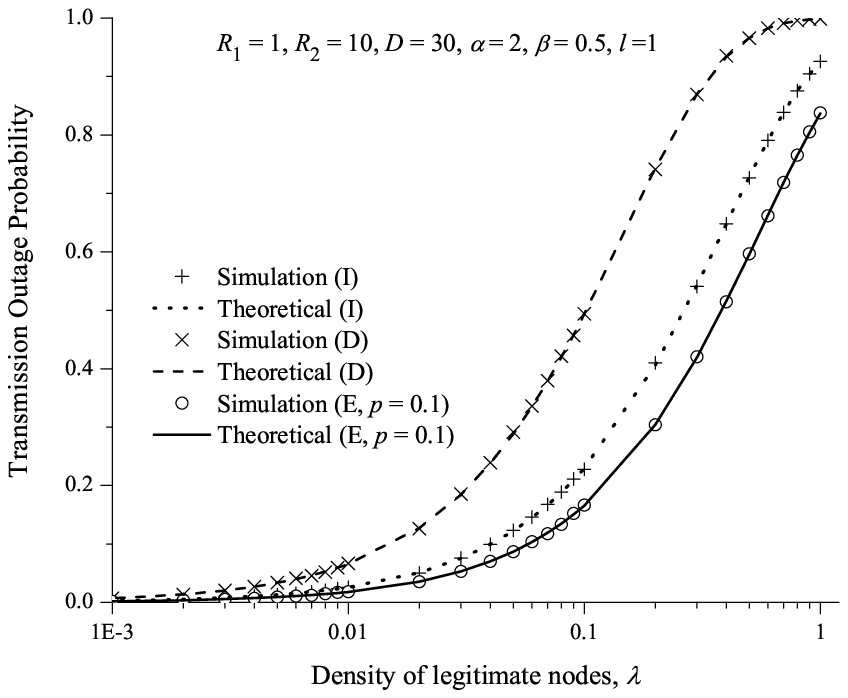}
\label{fig_TOP2_Sim}}
\hfil
\subfloat[SOP for $\alpha=2$]{\includegraphics[width=3in]{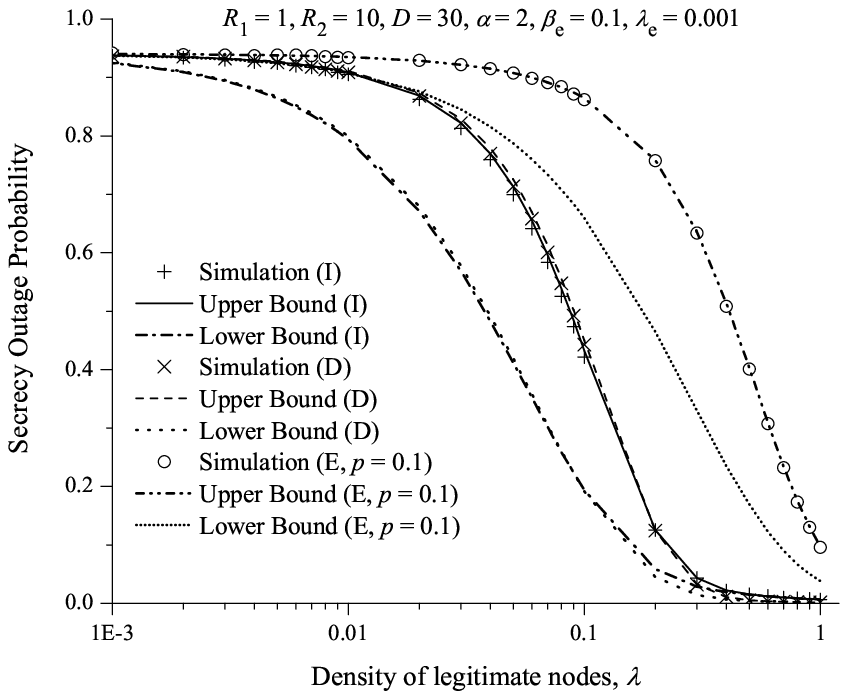}
\label{fig_SOP2_Sim}}
\vfil
\subfloat[TOP for $\alpha=4$]{\includegraphics[width=3in]{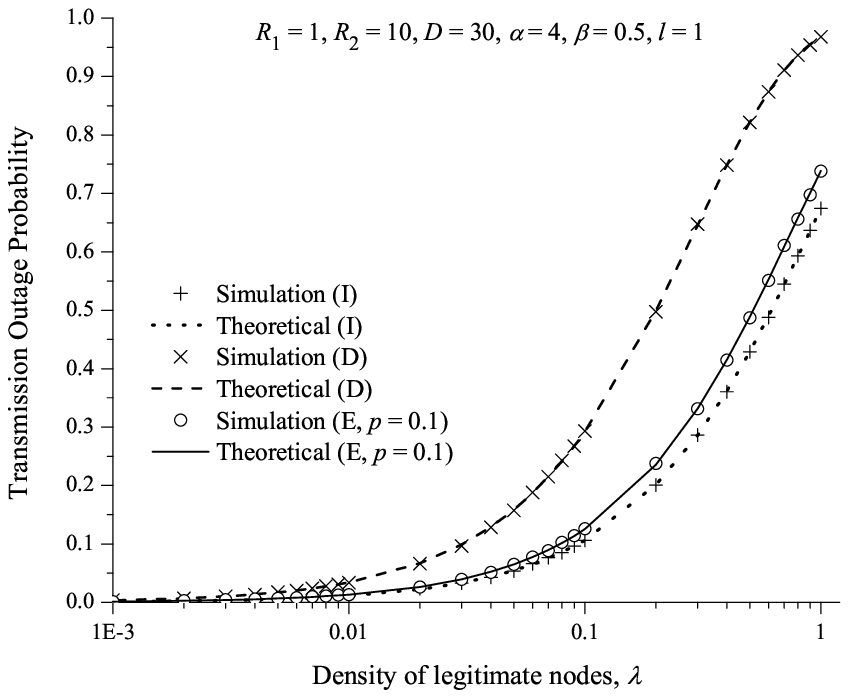}
\label{fig_TOP4_Sim}}
\hfil
\subfloat[SOP for $\alpha=4$]{\includegraphics[width=3in]{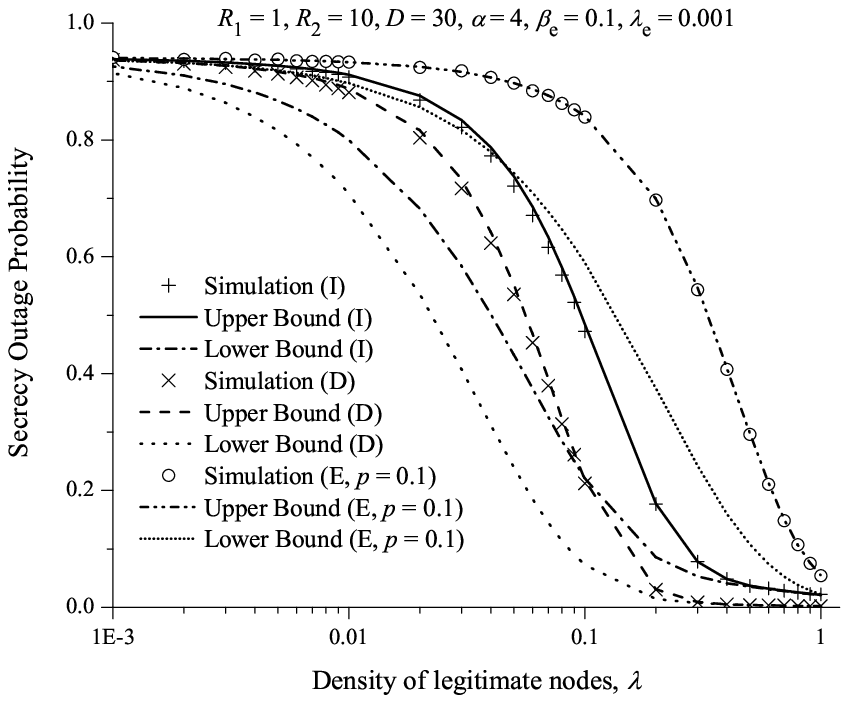}
\label{fig_SOP4_Sim}}
\caption{Simulation results vs. Theoretical results for TOP and SOP.}
\label{fig_sim}
\end{figure*}

\section{Numerical Results and Discussions}\label{sec_5}
In this section, we first conduct extensive simulations to verify the theoretical analysis of TOP and SOP. We then explore how the parameters of the friendship-based cooperative jamming scheme affect the TOP and SOP performances of the legitimate transmission. Finally, the impacts of the transmitter-receiver location and network size on the TOP and SOP performances are investigated. 
\subsection{Simulation Setting}
A simulator based on C++ was developed to simulate the PPPs $\Phi$ and $\Phi_E$, the friendship-based cooperative jamming model and the transmission process between the transmitter $o$ and receiver $y_0$, which is now available at \cite{Mdval4}. The PPP $\Phi$ ($\Phi_E$) is simulated by applying the method in \cite{Chiu2013}, where the first step is to generate a Poisson-distributed number $M$ with mean $\lambda \pi D^2$ (the mean is $\lambda_e \pi D^2$ for $\Phi_E$) and the second step is to distribute $M$ nodes uniformly over the network $\mathcal B(o,D)$. The total number of transmitter-receiver transmissions is fixed as $100000$ and the common transmit power is fixed as $1$. The TOP is calculated as the ratio of the number $n_{to}$ of transmissions with transmission outage to the total transmission number, i.e., $$\mathrm{TOP} = \frac{n_{to}}{100000}.$$ Similarly, The SOP is calculated as $$\mathrm{SOP} = \frac{n_{so}}{100000},$$ where $n_{so}$ is the number of transmissions with secrecy outage.

\subsection{Analysis Validation}

Extensive simulations have been conducted to verify the theoretical analysis of TOP and SOP. We considered the cases of $\alpha=2$ and $\alpha=4$ and examined how the TOP and SOP vary with the density of legitimate nodes $\lambda$ under three long-range jammer selection policies $\mathrm{E}$, $\mathrm{I}$ and $\mathrm{D}$.  For both path loss cases, the network radius was fixed as $D=30$ and the density of eavesdroppers was fixed as $\lambda_e=0.001$. For the friendship-based cooperative jamming scheme, the radius of the LFC was fixed as $R_1=1$, the outer radius of the LFA was fixed as $R_2=10$ and the selection probability in Policy $\mathrm{E}$ was set as $p=0.1$. The \textrm{SIR} thresholds were fixed as $\beta=0.5$ for the receiver $y_0$ and $\beta_e=0.1$ for eavesdroppers. The transmitter-receiver distance was set as $l=1$. The corresponding simulation results and theoretical results are summarized in Fig. \ref{fig_sim}.

Fig. \ref{fig_TOP2_Sim} and Fig. \ref{fig_TOP4_Sim} indicate clearly that the simulation results of TOP match nicely with the theoretical ones, so our theoretical results can be applied to model the TOP performance of the Poisson networks under Policy $\mathrm{E}$, Policy $\mathrm{I}$ and Policy $\mathrm{D}$ for the cases of $\alpha=2$ and $\alpha=4$. Fig. \ref{fig_SOP2_Sim} and Fig. \ref{fig_SOP4_Sim} indicate that the simulation results of SOP are very close to the corresponding theoretical upper bounds, while they are different from the lower bounds, so our theoretical upper bounds can serve as accurate approximations for the exact SOP of the legitimate transmission under Policy $\mathrm{E}$, Policy $\mathrm{I}$ and Policy $\mathrm{D}$ for the cases of $\alpha=2$ and $\alpha=4$. In the following, we mainly focus on the case of $\alpha=4$, as the behaviors of TOP and SOP for $\alpha=2$ and $\alpha=4$ are similar. In addition, we use the theoretical upper bounds on SOP in the discussions of the SOP performance.  

\begin{figure}[!t]
\centering
\includegraphics[width=3in]{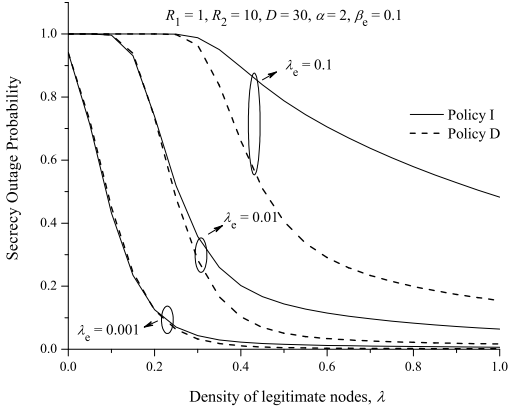}
\caption{SOP gap between Policy $\mathrm{I}$ and Policy $\mathrm{D}$ for $\alpha=2$.}
\label{fig_SOPGap}
\end{figure}

\subsection{TOP and SOP vs. Jamming Parameters}

We now explore how the TOP and SOP performances of the network vary with the parameters of the friendship-based cooperative jamming scheme with different long-range jammer selection policies.

\subsubsection {TOP and SOP vs. $\lambda$} 
We first examine the impact of the density of legitimate nodes $\lambda$ on the TOP and SOP performances.  It can be observed from Fig. \ref{fig_sim} that the TOP increases as $\lambda$ increases, while the SOP decreases as $\lambda$ increases under all policies $\mathrm{E}$, $\mathrm{I}$ and $\mathrm{D}$ for both $\alpha=2$ and $\alpha=4$. This is very intuitive since a larger sum interference can be generated in the network as $\lambda$ increases, degrading both the transmitter-receiver channel and eavesdropper channels. An interesting observation from Fig. \ref{fig_SOP2_Sim} indicates that Policy $\mathrm{I}$ and Policy $\mathrm{D}$ achieve almost the same SOP for $\alpha=2$ and $\lambda_e=0.001$. However, this is not the case for other settings of $\lambda_e$, as we can observe from Fig. \ref{fig_SOPGap}. Actually, as shown in Fig. \ref{fig_sim} and Fig. \ref{fig_SOPGap} that, in general, Policy $\mathrm{I}$ outperforms Policy $\mathrm{D}$ in terms of the TOP performance, while Policy $\mathrm{D}$ can ensure a better SOP performance than Policy $\mathrm{I}$. This is due to the following two reasons. The first one is that Policy $\mathrm{D}$ has much more long-range jammers than Policy $\mathrm{I}$, so it will generate more interference in the network, resulting in a better SOP performance but a worse TOP performance. The other reason is that the long-range jammers of Policy $\mathrm{D}$ are much closer to the transmitter than those of Policy $\mathrm{I}$. Notice that near (i.e., close to the transmitter) eavesdroppers dominate the behavior of SOP, so Policy $\mathrm{D}$ is more effective to suppress near eavesdroppers than Policy $\mathrm{I}$, achieving a better SOP performance.

Notice that in Fig. \ref{fig_sim}, the jammer selection probability of Policy $\mathrm{E}$ is fixed as $p=0.1$, which corresponds to a weak long-range jamming scenario. For the moderate long-range jamming scenario ($p=0.5$) and strong long-range jamming scenario ($p=1.0$), Fig. \ref{fig_TOPSOPvsP} shows TOP and SOP vs. $\lambda$ for $\alpha=4$. As shown in Fig. \ref{fig_TOPSOPvsP} that the behaviors of TOP and SOP are similar for different $p$. One can also observe from Fig. \ref{fig_TOPSOPvsP} that the TOP increases as $p$ increases, while the SOP decreases as $p$ increases. This indicates that we can flexibly control the TOP and SOP performances of Policy $\mathrm{E}$ by varying the long-range jammer selection probability $p$.
\begin{figure}[!t] 
\centering
\subfloat[TOP vs. $p$]{\includegraphics[width=3in]{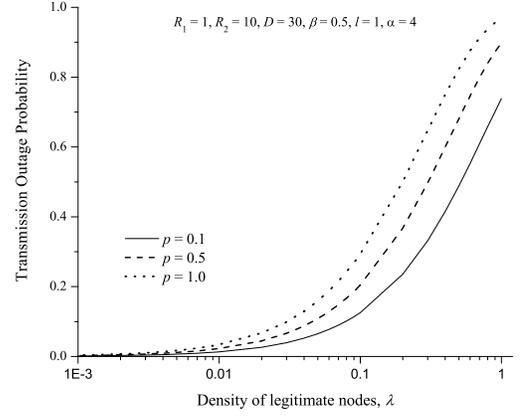}
\label{fig_TOPvsP}}
\vfil
\subfloat[SOP vs. $p$]{\includegraphics[width=3in]{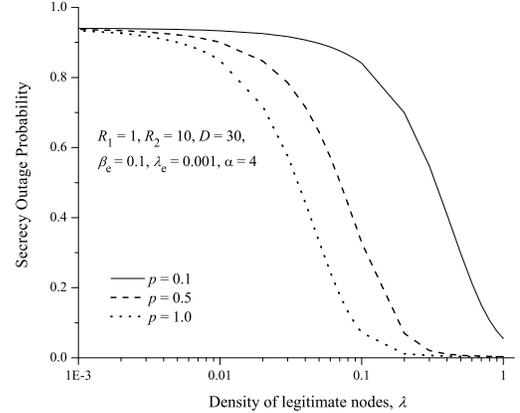}
\label{fig_SOPvsP}}
\caption{Impact of $p$ on TOP and SOP for Policy $\mathrm{E}.$}
\label{fig_TOPSOPvsP}
\end{figure}

\begin{figure}[!t]
\centering
\subfloat[TOP vs. $R_1$]{\includegraphics[width=3in]{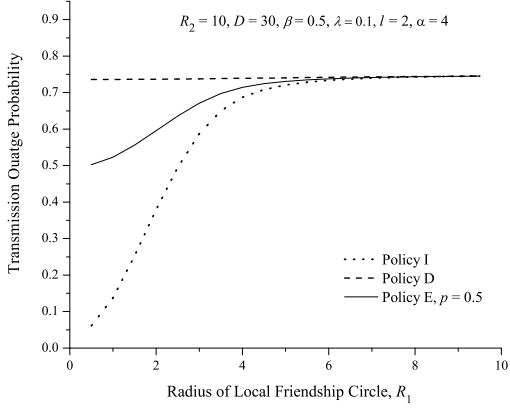}
\label{fig_TOPvsR1}}
\vfil
\subfloat[SOP vs. $R_1$]{\includegraphics[width=3in]{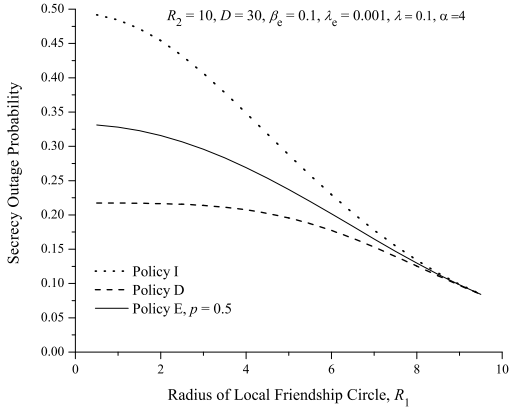}
\label{fig_SOPvsR1}}
\caption{Impact of $R_1$ on TOP and SOP.}
\label{fig_TOPSOPvsR1}
\end{figure}

\subsubsection{TOP and SOP vs. $R_1$}
We now investigate how the TOP and SOP performances are affected by the radius of LFC $R_1$, i.e., the inner radius of LFA. For the scenario of $R_2=10$, $D=30$, $\beta=0.5$, $\lambda=0.1$, $l=2$ and $\alpha=4$, Fig. \ref{fig_TOPvsR1} illustrates how the TOP varies with $R_1$ for Policy  $\mathrm{I}$, Policy  $\mathrm{D}$ and Policy $\mathrm{E}$ with $p=0.5$. We can see from Fig. \ref{fig_TOPvsR1} that the TOP first increases as $R_1$ increases, then saturates to a constant value and finally stays almost the same for Policy $\mathrm{I}$ and Policy $\mathrm{E}$. Actually, this is also the case for Policy $\mathrm{D}$. The increasing behavior of TOP is because that the total number of jammers increases as $R_1$ increases, although the number of long-range jammers decreases, which results in a larger sum interference in the network. The behavior that TOP of all policies saturates to a same constant is due to the fact that all policies finally reach to the same jamming pattern at the point of $R_1=R_2$. For the scenario of $R_2=10$, $D=30$, $\beta_e=0.1$, $\lambda_e=0.001$, $\lambda=0.1$ and $\alpha=4$, Fig. \ref{fig_SOPvsR1} shows how the SOP varies with $R_1$ for Policy  $\mathrm{I}$, Policy  $\mathrm{D}$ and Policy $\mathrm{E}$ with $p=0.5$. It can be observed from Fig. \ref{fig_SOPvsR1} that the SOP first decreases as $R_1$ increases, then saturates to a constant value and finally stays almost the same for all policies. This is due to the same reason as explained above.

\begin{figure}[!t]
\centering
\subfloat[TOP vs. $R_2$]{\includegraphics[width=3in]{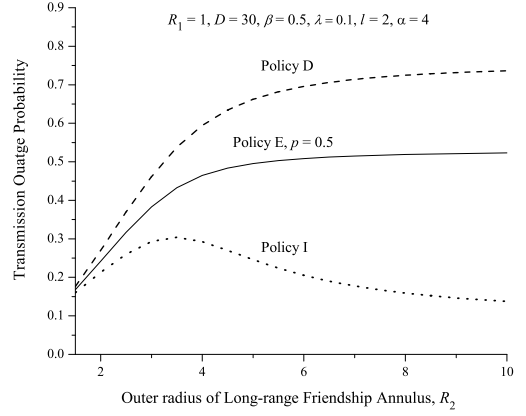}
\label{fig_TOPvsR2}}
\vfil
\subfloat[SOP vs. $R_2$]{\includegraphics[width=3in]{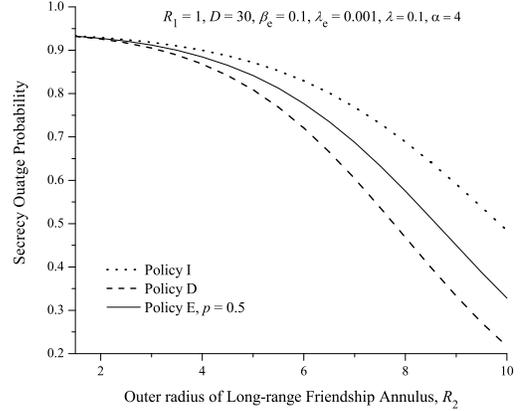}
\label{fig_SOPvsR2}}
\caption{Impact of $R_2$ on TOP and SOP.}
\label{fig_TOPSOPvsR2}
\end{figure}

\subsubsection{TOP and SOP vs. $R_2$}
Regarding the impact of the outer radius of LFA $R_2$ on the TOP performance, we show in Fig. \ref{fig_TOPvsR2} how the TOP varies with $R_2$ for Policy $\mathrm{I}$, Policy $\mathrm{D}$ and Policy $\mathrm{E}$ with $p=0.5$ under the settings of $R_1=1$, $D=30$, $\beta=0.5$, $\lambda=0.1$, $l=2$ and $\alpha=4$. As shown in Fig. \ref{fig_TOPvsR2} that the TOP of Policy $\mathrm{E}$ and Policy $\mathrm{D}$ always monotonically increases as $R_2$ increases, but this is not the case for Policy $\mathrm{I}$. The increasing behavior of TOP for all policies are because that the number of long-range jammers increases as $R_2$ increases, generating a larger sum interference in the network. The decreasing behavior of TOP for Policy $\mathrm{I}$ is due to that its long-range jammers are getting further away from the receiver as $R_2$ continues to increase, since these jammers are mainly located in a small annulus region near $R_2$, as we can deduce from (\ref{eqn_policyI}). For the impact of $R_2$ on the SOP performance, we illustrate in Fig. \ref{fig_SOPvsR2} SOP vs. $R_2$ for Policy $\mathrm{I}$, Policy $\mathrm{D}$ and Policy $\mathrm{E}$ with $p=0.5$ under the settings of $R_1=1$, $D=30$, $\beta_e=0.1$, $\lambda_e=0.001$, $\lambda=0.1$ and $\alpha=4$. As expected, we can observe from Fig.\ref{fig_SOPvsR2} that the SOP decreases as $R_2$ increases for all policies.

\subsection{SOP vs. Network Radius $D$}

We now explore how the SOP performance varies with the network radius $D$. For the scenario of $R_1=1$, $R_2=10$, $\beta_e=0.1$, $\lambda_e=0.001$, $\lambda=0.1$ and $\alpha=4$, Fig. \ref{fig_SOPvsD} illustrates how the SOP varies with $D$ for Policy $\mathrm{I}$, Policy $\mathrm{D}$ and Policy $\mathrm{E}$ with $p=0.5$. It is interesting to notice from Fig. \ref{fig_SOPvsD} that the SOP increases as the network radius $D$ increases and finally approaches $1$ for all policies, which is in accordance with Corollary \ref{corollary2}. Notice that this is somewhat counter-intuitive, since one might think that the SOP should finally approach a constant determined by $\beta_e$, $\lambda_e$, $\lambda$ and $\alpha$, like the result in \cite{Zhou2011} for infinite Poisson networks without the consideration of social friendships. Actually, since no jammers are for counteracting the eavesdroppers that are very far away from the transmitter in the friendship-based cooperative jamming scheme, the impacts of these eavesdroppers on the SOP cannot be simply neglected. 
\begin{figure}[!t]
\centering
\includegraphics[width=3in]{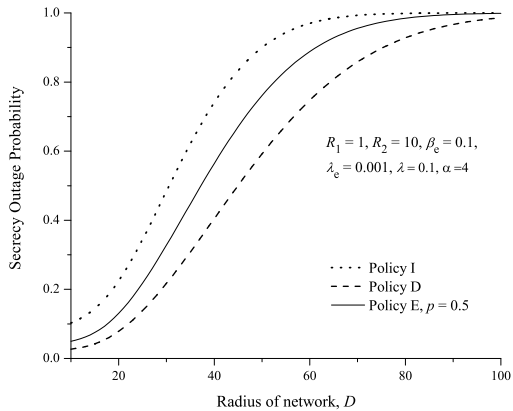}
\caption{Impact of network radius $D$ on SOP.}
\label{fig_SOPvsD}
\end{figure}

\subsection{TOP vs. Transmitter-Receiver Distance $l$}

\begin{figure}[!t]
\centering
\includegraphics[width=3in]{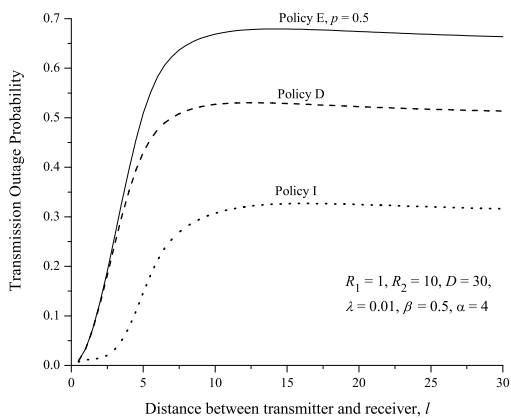}
\caption{Impact of transmitter-receiver distance $l$ on TOP.}
\label{fig_TOPvsL}
\end{figure}

To explore the impact of the transmitter-receiver distance $l$ on the TOP performance, we show in Fig. \ref{fig_TOPvsL} how the TOP varies with $l$ for Policy $\mathrm{I}$, Policy $\mathrm{D}$ and Policy $\mathrm{E}$ with $p=0.5$ under the settings of $R_1=1$, $R_2=10$, $D=30$, $\beta=0.5$, $\lambda=0.01$ and $\alpha=4$. We can observe from Fig. \ref{fig_TOPvsL} that the TOP increases as $l$ increases for all policies, which is intuitive since the received power decreases as $l$ increases. It is interesting to see from Fig. \ref{fig_TOPvsL} that the TOP finally saturates to a constant value for all policies. This is due to that as $l$ tends to infinity, the Laplace transform $\mathcal L_{I(y_0)}^{\mathrm{\Xi},\alpha}(\beta l^{\alpha})$ approaches a constant, which is easy to prove according to the proof of Corollary (\ref{corollary2}).

\section{Related Works}\label{sec_6}
Extensive research efforts have been devoted to the PHY-security based secure communications of Poisson networks without the consideration of social relationships, which can be roughly categorized according to the network scenarios they considered.

In general Poisson networks, the locations of eavesdroppers and legitimate nodes are usually modeled as independent and homogeneous PPPs with different intensities.  Some PHY-security properties of the networks were analyzed from the perspective of secrecy graph, like the secure connectivity, the maximum secrecy rate  and secrecy outage probability of a single link \cite{Pinto2012PartI, Pinto2012PartII}. Modeling the additional interfers as another independent homogeneous PPP, the authors in \cite{Rabbachin2015} explored some other PHY-security properties of the network, like secrecy rate density, secrecy rate outage density and secrecy throughput density. The dependence of the area spectral efficiency of Poisson networks on security and other parameters was studied in \cite{Zhou2011}.

In traditional cellular networks, base stations and mobile users  are usually modeled as independent and homogeneous PPPs. Recent efforts, such as \cite{HWang2013} and 
\cite{Geraci2014Cellular}, have been devoted to study the average secrecy rate achievable for a randomly located mobile user and the related probability of secrecy outage. For the cellular networks with D2D users, the authors in \cite{ChMa2015} modeled the locations of base stations, cellular users, D2D users and eavesdroppers as four independent and homogeneous PPPs, and studied the connection probabilities and secrecy probabilities of both the cellular and D2D links. 

It is notable that some recent works have also been reported on the study of PHY-security secure communications for other promising network scenarios, like cognitive networks \cite{Xu2016} and cognitive networks with D2D communications \cite{YWLiu2015ICC}.

\section{Conclusion}\label{sec_7}
This paper explored the physical layer security-based secure communications in a finite Poisson network with social friendships among nodes, for which a social friendship-based cooperative jamming scheme is proposed. The jamming scheme consists of a Local Friendship Circle (LFC) and a Long-range Friendship Annulus (LFA), where all legitimate nodes in the LFC serve as jammers, but the legitimate nodes in the LFA are selected as jammers through three location-based policies, namely, Policy $\mathrm{E}$, Policy $\mathrm{I}$ and Policy $\mathrm{D}$. To understand the security and reliability performances of the proposed jamming scheme, we analyzed its transmission outage probability (TOP) and secrecy outage probability (SOP) based on the Laplace transforms of the sum interference at any location in the network. The results in this paper indicated that, in general, Policy $\mathrm{I}$ outperforms Policy $\mathrm{D}$ in terms of the reliability performance, while Policy $\mathrm{D}$ can ensure a better security performance than Policy $\mathrm{I}$. Also, we can flexibly control the reliability and security performances of Policy $\mathrm{E}$ by varying its long-range jammer selection probability. Three other interesting observations can also be found from the results in this paper. The first one is that increasing the outer radius of the LFA beyond some threshold can improve both the reliability and security performances of the proposed jamming scheme. The second one is that as the network size tends to infinity, the transmission security can hardly be guaranteed, due to the fact that any eavesdropper located infinitely far away from the transmitter can still cause a non-zero SOP. This also gives rise to the last interesting observation, that is, even if the receiver is located infinitely far away from the transmitter , it can successfully receive the information with a non-zero probability in general.


%

\appendices
\section{Integral Identities} \label{sec_app1}
\begin{identity}\label{eqn_I1}
For $a,b\in\mathbb R$ and $a>|b|$, we have from \cite{Gradshteyn2000} and \cite{tanbourgi2012} 
\begin{IEEEeqnarray}{rCl}
\int_{0}^{\pi}\frac{\mathrm{d}\theta}{(a+b\cos \theta)^{n+1}}=\frac{\pi P_n(\frac{a}{\sqrt{a^2-b^2}})}{(a^2-b^2)^\frac{n+1}{2}},
\end{IEEEeqnarray}
where $P_n(\cdot)$ is the $n^{th}$-Legendre polynomial and $P_0(\cdot)=1$ .
\end{identity}

\begin{identity}\label{eqn_I2}
Let $a,b,c\in\mathbb R$ and $c>0$. Defining $Q=ct^2+bt+a$ and $\Delta=4ac-b^2$, we have from \cite{Gradshteyn2000} and \cite{tanbourgi2012} 
\begin{IEEEeqnarray}{rCl}
\int \frac{\mathrm{d}t}{\sqrt{Q}}
&=&\frac{1}{\sqrt{c}}\ln(2\sqrt{cQ}+2ct+b)~~~~~~~~[c>0]\nonumber\\
&=& \frac{1}{\sqrt{c}}\arcsinh \frac{2ct+b}{\sqrt{\Delta}}~~~~~[c>0,\Delta>0],
\end{IEEEeqnarray}
\end{identity}

\begin{identity}\label{eqn_I3}
For $m,n\in\mathbb Z$ and $Q=ct^2+bt+a$, we have from \cite{Gradshteyn2000} 
\begin{IEEEeqnarray}{rCl}
\int\frac{t^m}{\sqrt{Q^{2n+1}}}\mathrm{d}t&=&\frac{t^{m-1}}{(m-2n)c\sqrt{Q^{2n-1}}}\nonumber\\
&&-\frac{(2m-2n-1)b}{2(m-2n)c}\int\frac{t^{m-1}}{\sqrt{Q^{2n+1}}}\mathrm{d}t\nonumber\\
&&-\frac{(m-1)a}{(m-2n)c}\int\frac{t^{m-2}}{\sqrt{Q^{2n+1}}}\mathrm{d}t,
\end{IEEEeqnarray}
where $a,b,c\in\mathbb R$ and $c>0$.
\end{identity}

\section{Proof of Theorem \ref{theorem_LT_2}}\label{sec_app2}
For $\alpha = 2$, we can rewrite $B_\alpha$ in (\ref{eqn_B_alpha}) as
\begin{IEEEeqnarray}{rCl}
B_2&=&2\int_{0}^{R_1}\int_{0}^{\pi}\frac{sr\mathrm{d}\theta\mathrm{d}r}{s+r^2+||y||^2-2r||y||\cos\theta}.
\end{IEEEeqnarray}
Applying Identity \ref{eqn_I1} in Appendix \ref{sec_app1}, we have 
\begin{IEEEeqnarray}{rCl}
B_2&=&\pi s\int_{0}^{R_1}\frac{2r\mathrm{d}r}{\sqrt{r^4+2(s-||y||^2)r^2+(s+||y||^2)^2}}\nonumber\\
&\overset{(c)}=&\pi s \int_{0}^{R_1^2}\frac{\mathrm{d}t}{\sqrt{(t^2+2(s-||y||^2)t+(s+||y||^2)^2}},
\end{IEEEeqnarray}
where $(c)$ follows from substituting $r^2$ with $t$. We then apply Identity \ref{eqn_I2} in Appendix \ref{sec_app1} and substitute $t$ with $r^2$ to obtain
\begin{IEEEeqnarray}{rCl}\label{eqn_B2}
B_2&=&\pi s\left(\arcsinh\frac{s+R_1^2-||y||^2}{2||y||\sqrt{s}}-\ln\frac{\sqrt{s}}{||y||}\right).
\end{IEEEeqnarray}
Similarly, applying Identity \ref{eqn_I1}, we can rewrite $C_\alpha$ in (\ref{eqn_C_alpha}) as
\begin{IEEEeqnarray}{rCl}
C_2&=&\pi s\int_{R_1}^{R_2}\frac{2rP(r)\mathrm{d}r}{\sqrt{r^4+2(s-||y||^2)r^2+(s+||y||^2)^2}}.
\end{IEEEeqnarray}
For Policy $\mathrm{E}$, $P(r) = p$. Then,
\begin{IEEEeqnarray}{rCl}\label{eqn_C2_E}
C_2&=&p\pi s\arcsinh \frac{s+r^2-||y||^2}{2||y||\sqrt{s}}\bigg|_{r=R_1}^{R_2}.
\end{IEEEeqnarray}
Substituting (\ref{eqn_C2_E}) and (\ref{eqn_B2}) into (\ref{eqn_A}), and then substituting (\ref{eqn_A}) into (\ref{eqn_LT_general}) yields (\ref{eqn_LT_E2}).
$P(r)$ can be written as $P(r)=u+vr^2$, where $u = -\frac{R_1^2}{R_2^2-R_1^2}$, $v = \frac{1}{R_2^2-R_1^2}$ for Policy \textup{I}, and $u = \frac{R_2^2}{R_2^2-R_1^2}$, $v = -\frac{1}{R_2^2-R_1^2}$ for Policy \textup{D}. Hence, 
\begin{IEEEeqnarray}{rCl}
C_2&=&\pi s\int_{R_1}^{R_2}\frac{2r(u+vr^2)\mathrm{d}r}{\sqrt{r^4+2(s-||y||^2)r^2+(s+||y||^2)^2}}\nonumber\\
&=&\pi s\int_{R_1^2}^{R_2^2}\frac{(u+vt)\mathrm{d}t}{\sqrt{(t^2+2(s-||y||^2)t+(s+||y||^2)^2}}\nonumber\\
&=&\pi s\bigg[u\int_{R_1^2}^{R_2^2}\frac{\mathrm{d}t}{\sqrt{(t^2+2(s-||y||^2)t+(s+||y||^2)^2}}\nonumber\\
&&+v\int_{R_1^2}^{R_2^2}\frac{t\mathrm{d}t}{\sqrt{(t^2+2(s-||y||^2)t+(s+||y||^2)^2}}\bigg]\nonumber\\
&\overset{(d)}=&\pi s\bigg[(u-vs+v||y||^2)\arcsinh \frac{s+t-||y||^2}{2||y||\sqrt{s}}\nonumber\\
&&+v \sqrt{(t^2+2(s-||y||^2)t+(s+||y||^2)^2}\bigg]\bigg|_{t=R_1^2}^{R_2^2},
\end{IEEEeqnarray}
where $(d)$ follows from applying Identity \ref{eqn_I2} and Identity \ref{eqn_I3}. Substituting $t$ with $r^2$, we have 
\begin{IEEEeqnarray}{rCl}\label{eqn_C2_I_D}
C_2&=&\pi s\bigg[(u-vs+v||y||^2)\arcsinh \frac{s+r^2-||y||^2}{2||y||\sqrt{s}} \\
&&+v \sqrt{(r^4+2(s-||y||^2)r^2+(s+||y||^2)^2}\bigg]\bigg|_{r=R_1}^{R_2}.\nonumber
\end{IEEEeqnarray}

Finally, we substitute (\ref{eqn_B2}) and (\ref{eqn_C2_I_D}) with  into (\ref{eqn_A}), and then substitute (\ref{eqn_A}) into (\ref{eqn_LT_general}) to obtain (\ref{eqn_LT_I_D2}).

\section{Proof of Theorem \ref{theorem_LT_4}}\label{sec_app3}
For $\alpha=4$, we can rewrite $B_\alpha$ in (\ref{eqn_B_alpha}) as
\begin{IEEEeqnarray}{rCl}
B_4&=&2\int_{0}^{R_1}\int_{0}^{\pi}\frac{sr\mathrm{d}\theta\mathrm{d}r}{s+(r^2+||y||^2-2r||y||\cos\theta)^{2}}\nonumber\\
&=&2\int_{0}^{R_1}\frac{\sqrt{s}r}{2 i}\int_{0}^{\pi}\frac{\mathrm{d}\theta\mathrm{d}r}{(r^2+||y||^2-2r||y||\cos\theta-i\sqrt{s})}\nonumber\\
&&-\frac{\mathrm{d}\theta\mathrm{d}r}{(r^2+||y||^2-2r||y||\cos\theta+i\sqrt{s})}\nonumber\\
&\overset{(e)}=&\frac{\pi\sqrt{s}}{2i}\int_{0}^{R_1}\frac{2r\mathrm{d}r}{\sqrt{\mathcal C_1}}-\frac{2r\mathrm{d}r}{\sqrt{\mathcal C_2}}\nonumber\\
&\overset{(f)}=&\frac{\pi\sqrt{s}}{2i}\ln \frac{\sqrt{\mathcal C_1}+r^2-(i\sqrt{s}+||y||^2)}{\sqrt{\mathcal C_2}+r^2+(i\sqrt{s}-||y||^2)}\bigg|_{r=0}^{R_1},
\end{IEEEeqnarray}
where $(e)$ follows from applying Identity \ref{eqn_I1}, $(f)$ follows from applying Identity \ref{eqn_I2},
\begin{IEEEeqnarray}{rCl}\label{eqn_C1}
\mathcal C_1&=&(r^2-||y||^2)^2-s-2i\sqrt{s}(r^2+||y||^2)\nonumber,
\end{IEEEeqnarray}
and $\mathcal C_2=\mathcal C_1^*$ is the complex conjugate of $\mathcal C_1$. 
Now, we rewrite $\mathcal C_1$ as 
\begin{IEEEeqnarray}{rCl}
\mathcal C_1=(\eta -i\psi )^2=\eta^2-\psi^2-2i\eta\psi,
\end{IEEEeqnarray}  
for some real-valued functions $\eta(r,s,||y||)$ and $\psi(r,s,||y||)$. We can then establish the following equation system
\begin{IEEEeqnarray}{rCl}
\left\{\begin{matrix}
\eta^2-\psi ^2&=&(r^2-||y||^2)^2-s\\ 
\eta \psi&=&\sqrt{s}(r^2+||y||^2).
\end{matrix}\right.
\end{IEEEeqnarray} 
(\ref{eqn_eta}) and (\ref{eqn_psi}) then follow from solving the above equation system. For simplicity of notation, we also use $\eta$ and $\psi$ to represent $\eta(r,s,||y||)$ and $\psi(r,s,||y||)$, respectively. Given $\mathcal C_1$ as in (\ref{eqn_C1}), 
\begin{IEEEeqnarray}{rCl}\label{eqn_B4}
B_4&=&\frac{\pi\sqrt{s}}{2i}\ln \frac{\eta+r^2-||y||^2-i(\sqrt{s}+\psi)}{\eta+r^2-||y||^2+i(\sqrt{s}+\psi)}\bigg|_{r=0}^{R_1}\\
&=&\frac{\pi\sqrt{s}}{2i}\ln \frac{1-i\frac{\sqrt{s}+\psi}{\eta+r^2-||y||^2}}{1+i\frac{\sqrt{s}+\psi}{\eta+r^2-||y||^2}}\bigg|_{r=0}^{R_1}\nonumber\\
&=&-\pi\sqrt{s}\arctan\frac{\sqrt{s}+\psi}{\eta+r^2-||y||^2}\bigg|_{r=0}^{R_1}\nonumber\\
&\overset{(g)}=&\pi\sqrt{s}\left(\frac{\pi}{2}-\arctan\frac{\sqrt{s}+\psi(R_1,s,||y||)}{\eta(R_1,s,||y||)+R_1^2-||y||^2}\right),\nonumber
\end{IEEEeqnarray}
where $(g)$ follows from 
\begin{IEEEeqnarray}{rCl}
&&\lim_{r\rightarrow 0}\arctan\frac{\sqrt{s}+\psi(r,s,||y||)}{\eta(r,s,||y||)+r^2-||y||^2}\nonumber\\
&&=\lim_{r\rightarrow 0}\arctan\frac{\sqrt{s}+\sqrt{2s}}{||y||^2+r^2-||y||^2}\nonumber\\
&&=\arctan \infty=\frac{\pi}{2}.
\end{IEEEeqnarray}
Similarly, applying Identity \ref{eqn_I1}, we can rewrite $C_\alpha$ in (\ref{eqn_C_alpha}) as 
\begin{IEEEeqnarray}{rCl}
C_4&=&\frac{\pi\sqrt{s}}{2i}\int_{R_1}^{R_2}\frac{2rP(r)\mathrm{d}r}{\sqrt{\mathcal C_1}}-\frac{2rP(r)\mathrm{d}r}{\sqrt{\mathcal C_2}},
\end{IEEEeqnarray}
For Policy $\mathrm{E}$, $P(r)=p\in[0,1]$. Then, 
\begin{IEEEeqnarray}{rCl}\label{eqn_C4_E}
C_4&=&-p\pi\sqrt{s}\arctan\frac{\sqrt{s}+\psi(r,s,||y||)}{\eta(r,s,||y||)+r^2-||y||^2}\bigg|_{r=R_1}^{R_2}.
\end{IEEEeqnarray} 
Substituting (\ref{eqn_C4_E}) and (\ref{eqn_B4}) into (\ref{eqn_A}) and then substituting (\ref{eqn_A}) into (\ref{eqn_LT_general}) yields (\ref{eqn_LT_E4}).
$P(r)$ can be written as $P(r)=u+vr^4$, where $u = -\frac{R_1^4}{R_2^4-R_1^4}$, $v = \frac{1}{R_2^4-R_1^4}$ for Policy \textup{I}, and $u = \frac{R_2^4}{R_2^4-R_1^4}$, $v = -\frac{1}{R_2^4-R_1^4}$ for Policy \textup{D} . Hence,
\begin{IEEEeqnarray}{rCl}\label{eqn_C4_I_D}
C_4&=&\frac{\pi\sqrt{s}}{2i}\int_{R_1}^{R_2}\frac{2r(u+vr^4)\mathrm{d}r}{\sqrt{\mathcal C_1}}-\frac{2r(u+vr^4)\mathrm{d}r}{\sqrt{\mathcal C_2}}\mathrm{d}r\nonumber\\
&\overset{(h)}=&\frac{\pi\sqrt{s}}{2i}\int_{R_1}^{R_2}\frac{(u+vt^2)\mathrm{d}t}{\sqrt{t^2-2(i\sqrt{s}+||y||^2)t+(||y||^2-i\sqrt{s})^2}}\nonumber\\
&&-\frac{(u+vt^2)\mathrm{d}t}{\sqrt{t^2+2(i\sqrt{s}-||y||^2)t+(||y||^2+i\sqrt{s})^2}},
\end{IEEEeqnarray}
where $(h)$ follows from substituting $r^2$ with $t$. Next, we have
\begin{IEEEeqnarray}{rCl}\label{eqn_C4_I_D_1}
&&\int\frac{(u+vt^2)\mathrm{d}t}{\sqrt{t^2-2(i\sqrt{s}+||y||^2)t+(||y||^2-i\sqrt{s})^2}}\nonumber\\
&&=u\int\frac{\mathrm{d}t}{\sqrt{t^2-2(i\sqrt{s}+||y||^2)t+(||y||^2-i\sqrt{s})^2}}\nonumber\\
&&~~+v\int\frac{t^2\mathrm{d}t}{\sqrt{t^2-2(i\sqrt{s}+||y||^2)t+(||y||^2-i\sqrt{s})^2}}\nonumber\\
&&\overset{(i)}=\frac{v}{2}(r^2+3||y||^2+3i\sqrt{s})(\eta-i\psi)\nonumber\\
&&~~~+(u+v||y||^4-vs+i4v\sqrt{s}||y||^2)\nonumber\\
&&~~~\times\ln \left[\sqrt{\mathcal C_1}+r^2-(i\sqrt{s}+||y||^2)\right],
\end{IEEEeqnarray}
where $(i)$ follows from applying Identity \ref{eqn_I3} in Appendix \ref{sec_app1} and substituting $t$ with $r^2$. Similarly, we have
\begin{IEEEeqnarray}{rCl}\label{eqn_C4_I_D_2}
&&\int\frac{(u+vt^2)\mathrm{d}t}{\sqrt{t^2+2(i\sqrt{s}-||y||^2)t+(||y||^2+i\sqrt{s})^2}}\nonumber\\
&&=\frac{v}{2}(r^2+3||y||^2-3i\sqrt{s})(\eta+i\psi)\nonumber\\
&&~~~+(u+v||y||^4-vs-i4v\sqrt{s}||y||^2)\nonumber\\
&&~~~\times\ln \left[\sqrt{\mathcal C_2}+r^2-(i\sqrt{s}+||y||^2)\right].
\end{IEEEeqnarray}

Thus, substituting (\ref{eqn_C4_I_D_1}) and (\ref{eqn_C4_I_D_2}) into (\ref{eqn_C4_I_D}) and then conducting some algebraic manipulations yields
\begin{IEEEeqnarray}{rCl}\label{eqn_C4_I_D_final}
C_4&=&2\pi vs||y||^2\ln\bigg[(\eta(r,s,||y||)+r^2-||y||^2)^2\nonumber\\
&&+(\sqrt{s}+\psi(r,s,||y||))^2\bigg]-\pi\sqrt{s}\Bigg\{\frac{v}{2}\bigg[(r^2+3||y||^2)\nonumber\\
&&\times\psi(r,s,||y||)-3\sqrt{s}\eta(r,s,||y||)\bigg]+(u+v||y||^4-vs)\nonumber\\
&&\times \arctan\frac{\sqrt{s}+\psi(r,s,||y||)}{\eta(r,s,||y||)+r^2-||y||^2}\Bigg\}\Bigg|_{r=R_1}^{R_2}.
\end{IEEEeqnarray} 

Finally, we substitute (\ref{eqn_B4}) and (\ref{eqn_C4_I_D_final}) into (\ref{eqn_A}), and then substitute (\ref{eqn_A}) into (\ref{eqn_LT_general}) to obtain (\ref{eqn_LT_I_D4}).

\section{Probability Density Function of $R_{z^*}$}\label{sec_app4}
The CCDF $\bar{F}_{R_{z^*}}(r_{e^*})$ of the random distance $R_{z^*}$ equals the probability that no eavesdroppers are in $\mathcal B(o,r_{e^*})$ for $0 \leq r_{e^*}\leq D$.
Hence, the CDF of $R_{z^*}$ is given by
\begin{IEEEeqnarray}{rCl}
F_{R_{z^*}}(r_{e^*})&=&1-\mathbb P\left(\Phi_E(\mathcal B(o,r_{e^*}))=0\right)\nonumber\\
&=&1-\sum_{n=0}^{\infty}\mathbb P\left(\Phi_E(\mathcal B(o,r_{e^*}))=0\big|\Phi_E(\mathcal B(o,D))=n\right)\nonumber\\
&&\times\mathbb P(\Phi_E(\mathcal B(o,D))=n)\nonumber\\
&=&1-\sum_{n=0}^{\infty}\left(1-\frac{r_{e^*}^2}{D^2}\right)^n\frac{(\lambda_e\pi D^2)^n\mathrm{exp}(-\lambda_e \pi D^2)}{n!}\nonumber\\
&=&1-\mathrm{exp}(-\lambda_e \pi D^2)\sum_{n=0}^{\infty}\left(1-\frac{r_{e^*}^2}{D^2}\right)^n\frac{(\lambda_e\pi D^2)^n}{n!}\nonumber\\
&=&1-\mathrm{exp}(-\lambda_e \pi D^2)\mathrm{exp}\left[\left(1-\frac{r_{e^*}^2}{D^2}\right)\lambda_e\pi D^2\right]\nonumber\\
&=&1-\mathrm{exp}(-\lambda_e \pi r_{e^*}^2),
\end{IEEEeqnarray}
for $0 \leq r_{e^*}\leq D$. Therefore, the pdf of $R_{z^*}$ is given by
\begin{IEEEeqnarray}{rCl}
f_{R_{z^*}}(r_{e^*})=
\left\{\begin{matrix}
 2\lambda_e\pi r_{e^*}\mathrm{exp}(-\lambda_e \pi r_{e^*}^2), & 0\leq r_{e^*}\leq D\nonumber\\ 
 0,& \mathrm{otherwise} 
\end{matrix}\right..
\end{IEEEeqnarray}

\section{Proof of Corollary \ref{corollary2}}\label{sec_app5}
Consider an annulus with inner radius $D-\epsilon$ and outer radius $D$, where $\epsilon>0$ is a constant. The basic idea is to first prove that as $D\rightarrow \infty$, the probability of secrecy outage caused by any eavesdropper $z$ in the annulus is above a constant, and then prove that the expected number of eavesdroppers in the annulus tends to infinity, as $D\rightarrow \infty$. 

We first prove the former part.  For any eavesdropper $z$ in the annulus, we can find a constant $\epsilon'$ ($0<\epsilon'<\epsilon$), such that $||z||=D-\epsilon'$. The probability of secrecy outage caused by the eavesdropper $z$ is then bounded from below by that derived for the case where all legitimate nodes in $\mathcal A_1\bigcup \mathcal A_2$ serve as jammers, i.e.,
\begin{IEEEeqnarray}{rCl}
&&\mathbb P(\mathrm{SIR}_{z}\geq \beta_e)\\
&&=\mathcal L_{I(z)}^{\Xi,\alpha}(\beta_e(D-\epsilon')^\alpha)\nonumber\\
&&\geq\mathrm{exp}\Bigg\{-\lambda\int_{0}^{R_2}\int_{0}^{\pi}\nonumber\\
&&\frac{2\beta_e(D-\epsilon')^\alpha r\mathrm{d}\theta\mathrm{d}r}{\beta_e(D-\epsilon')^\alpha+(r^2+(D-\epsilon')^2-2r(D-\epsilon')\cos\theta)^{\frac{\alpha}{2}}}\Bigg\}.\nonumber
\end{IEEEeqnarray}
As $D\rightarrow \infty$, 
\begin{IEEEeqnarray}{rCl}\label{eqn_LT_infty}
&&\lim_{D\rightarrow \infty}\mathbb P(\mathrm{SIR}_{z}\geq \beta_e)\nonumber\\
&&\geq\lim_{D\rightarrow \infty}\mathrm{exp}\Bigg\{-\lambda\int_{0}^{R_2}\int_{0}^{\pi}\nonumber\\
&&\frac{2\beta_e(D-\epsilon')^\alpha r\mathrm{d}\theta\mathrm{d}r}{\beta_e(D-\epsilon')^\alpha+(r^2+(D-\epsilon')^2-2r(D-\epsilon')\cos\theta)^{\frac{\alpha}{2}}}\Bigg\}\nonumber\\
&&=\mathrm{exp}\Bigg\{-\lambda\lim_{D\rightarrow \infty}\int_{0}^{R_2}\int_{0}^{\pi}\nonumber\\
&&\frac{2\beta_e(D-\epsilon')^\alpha r\mathrm{d}\theta\mathrm{d}r}{\beta_e(D-\epsilon')^\alpha+(r^2+(D-\epsilon')^2-2r(D-\epsilon')\cos\theta)^{\frac{\alpha}{2}}}\Bigg\}\nonumber\\
&&=\mathrm{exp}\Bigg\{-\lambda\lim_{D\rightarrow \infty}\int_{0}^{R_2}\int_{0}^{\pi}\nonumber\\
&&\frac{2\beta_e r\mathrm{d}\theta\mathrm{d}r}{\beta_e+\left[\frac{r^2}{(D-\epsilon')^2}+1-\frac{2r}{(D-\epsilon')}\cos\theta\right]^{\frac{\alpha}{2}}}\Bigg\}\nonumber\\
&&=\mathrm{exp}\left(-\lambda\int_{0}^{R_2}\int_{0}^{\pi}\frac{2\beta_e r\mathrm{d}\theta\mathrm{d}r}{\beta_e+1}\right)\nonumber\\
&&=\mathrm{exp}\left(-\lambda \pi R_2^2 \frac{\beta_e}{\beta_e+1}\right).
\end{IEEEeqnarray}

We now prove the latter part. According to the property of homogeneous PPP, the expected number of eavesdroppers in this annulus is 
\begin{IEEEeqnarray}{rCl}
\lambda_e\pi(D^2-(D-\epsilon)^2)=\lambda_e\pi\epsilon(2D-\epsilon).
\end{IEEEeqnarray}
It is easy to see that $\lim_{D\rightarrow\infty}\lambda_e\pi\epsilon(2D-\epsilon)=\infty$. The probability of secrecy outage caused by the eavesdroppers in the annulus can then be approximated by 
\begin{IEEEeqnarray}{rCl}
1-\left(1-\mathrm{exp}\left(-\lambda \pi R_2^2 \frac{\beta_e}{\beta_e+1}\right)\right)^\infty=1,
\end{IEEEeqnarray}
which completes the proof.

%
%
%
%



%

\bibliographystyle{IEEEtran}
\bibliography{myIEEEref}

%
\begin{IEEEbiography}[]{Yuanyu Zhang}
received his B.S. degree in Software Engineering from Xidian University in 2011 and M.S. degrees in Computer Science from Xidian University in 2014. He is currently working towards a Ph.D. degree at the School of Systems Information Science at Future University Hakodate. His research interests include the physical layer security of wireless communications, and performance modeling and evaluation of wireless networks.
\end{IEEEbiography}

\begin{IEEEbiography}[]{Yulong Shen}
received the B.S. and M.S. degrees in Computer Science and Ph.D. degree in Cryptography from Xidian University, Xian, China, in 2002, 2005, and 2008, respectively. He is currently a Professor at the School of Computer Science and Technology, Xidian University, China. He is also an associate director of the Shaanxi Key Laboratory of Network and System Security and a member of the State Key Laboratory of Integrated Services networks Xidian University, China. He has also served on the technical program committees of several international conferences, including ICEBE, INCoS, CIS and SOWN. His research interests include Wireless network security and cloud computing security.
\end{IEEEbiography}

\begin{IEEEbiography}[]{Hua Wang}
received his PhD degree from the University of Southern Queensland, Australia. He is now a full time Professor at Victoria University. He was a professor at the University of Southern Queensland before he joined Victoria University.  Hua has more than ten years teaching and working experience in Applied Informatics at both enterprise and university.  He has expertise in electronic commerce, business process modeling and enterprise architecture.  As an Chief Investigator, three Australian Research Council (ARC) Discovery grants have been awarded since 2006, and 155 peer reviewed scholar papers have been published. Six PhD students have already graduated under his principal supervision.
\end{IEEEbiography}
%

\begin{IEEEbiography}[]{Xiaohong Jiang}
Dr.Xiaohong Jiang received his B.S., M.S. and Ph.D degrees in 1989, 1992, and 1999 respectively, all from Xidian University, China. He is currently a full professor of Future University Hakodate, Japan. Before joining Future University, Dr.Jiang was an Associate professor, Tohoku University, from Feb.2005 to Mar.2010. Dr. Jiang’s research interests include computer communications networks, mainly wireless networks and optical networks, network security, routers/switches design, etc. He has published over 260 technical papers at premium international journals and conferences, which include over 50 papers published in top IEEE journals and top IEEE conferences, like IEEE/ACM Transactions on Networking, IEEE Journal of Selected Areas on Communications, IEEE Transactions on Parallel and Distributed Systems, IEEE INFOCOM.  Dr. Jiang was the winner of the Best Paper Award of IEEE HPCC 2014, IEEE WCNC 2012, IEEE WCNC 2008, IEEE ICC 2005-Optical Networking Symposium, and IEEE/IEICE HPSR 2002. He is a Senior Member of IEEE, a Member of ACM and IEICE.
\end{IEEEbiography}
%
%
%




\end{document}